\documentclass[11pt]{amsart}
\usepackage{amsfonts}
\usepackage{graphicx,color}
\usepackage{amsthm}
\usepackage{amssymb,amsmath}
\usepackage{graphicx}
\usepackage{psfrag}

\title{\bf Refinements of Levenshtein bounds in $q$-ary Hamming spaces }

\date{\today}

\newtheorem{theorem}{Theorem}[section]
\newtheorem{lemma}[theorem]{Lemma}
\newtheorem{corollary}[theorem]{Corollary}

\newtheorem{conjecture}[theorem]{Conjecture}

\theoremstyle{definition}

\newtheorem{example}[theorem]{Example}
\newtheorem{remark}[theorem]{Remark}

\author[P. Boyvalenkov]{Peter Boyvalenkov $^\dagger$}
\address{Institute of Mathematics and Informatics, Bulgarian Academy
	of Sciences, 8 G Bonchev Str., 1113 Sofia, Bulgaria \\ and Faculty
	of Mathematics and Natural Sciences, South-Western University,
	Blagoevgrad, Bulgaria.  } \email{peter@math.bas.bg}
\thanks{\noindent $^\dagger$ The research of this author was
	supported, in part, by a Bulgarian NSF contract DN02/2-13.12.2016.}

\author[D. Danev]{Danyo Danev $^{\dagger \dagger}$}
\address{Department of Electrical Engineering, Link\"oping University
	SE-58183 Link\"oping, Sweden} \email{danyo.danev@liu.se} \thanks{\noindent
	$^{\dagger \dagger}$ The research of this author was supported, in part, by  the Swedish Research Council (VR) and ELLIIT.}

\author[M. Stoyanova]{Maya Stoyanova$^*$}
\address{Faculty of Mathematics and Informatics, Sofia University, 5
	James Bourchier Blvd., 1164 Sofia, Bulgaria}
\email{stoyanova@fmi.uni-sofia.bg} \thanks{\noindent $^*$ The research
	of this author was supported, in part, by a Bulgarian NSF contract DN02/2-13.12.2016.}

\begin{document}

	\maketitle

\centerline{Dedicated to the memory of Professor Vladimir Levenshtein (1935-2017)}

	\begin{abstract}
		We develop refinements of the Levenshtein bound in $q$-ary Hamming
		spaces by taking into account the discrete nature of the distances
		versus the continuous behavior of certain parameters used by
		Levenshtein.  The first relevant cases are investigated in detail and
		new bounds are presented. In particular, we derive generalizations
		and $q$-ary analogs of a MacEliece bound. We provide evidence
		that our approach is as good as the complete linear programming and
		discuss how faster are our calculations.  Finally, we present a
		table with parameters of codes which, if exist, would attain our bounds.
	\end{abstract}

	
	\section{Introduction}
	
	Let $Q=\{0,1,\ldots,q-1\}$ be the alphabet of $q$ symbols and $H(n,q)$
	be the set of all $q$-ary vectors $x=(x_1,x_2,\ldots,x_n)$ over
	$Q$. The Hamming distance $d(x,y)$ between points
	$x=(x_1,x_2,\ldots,x_n)$ and $y=(y_1,y_2,\ldots,y_n)$ from $H(n,q)$ is
	equal to the number of coordinates in which they differ. A non-empty
	set $C \subset H(n,q)$ is called a code. The investigation of the
	connections between the codelength $n$, cardinality $|C|$ and the
	minimum distance $d=d(C)=\min\{ d(x,y), \ x \neq y \in C\}$ is of
	great importance in Coding Theory.
	
	The spaces $H(n,q)$ are sometimes considered as {\em polynomial metric}
	spaces
	(cf. \cite{DL,Lev92,Lev}), where using ''inner''
	products $\langle x,y \rangle := 1-\frac{2d(x,y)}{n}$ instead of
	distances is very convenient. We define $T_n=\{t_0,t_1,\ldots,t_n\}$, where
	$t_i :=-1+\frac{2i}{n}$, $i=0,1,\ldots,n$, as the set of all possible inner
	products.
	
	Let $s \in T_n$ and
	\[
	A_q(n,s):=\max\{|C| \colon C \subset H(n,q), s(C) \leq s\},
	\]
	where $s(C)=\max \{\langle x,y \rangle, \ x \neq y \in C\} $, be the
	maximal possible cardinality of a code in $H(n,q)$ of prescribed
	maximal inner product $s$. In Coding Theory this quantity is usually
	denoted by $A_q(n,d)$, where $s=1-\frac{2d}{n}$ ($=t_{n-d}$) and $d$ is
	the minimum distance of $C$ (so we have replaced the condition $d(x,y)
	\geq d$ by $\langle x,y \rangle \leq s$).
	
	Levenshtein (cf. \cite{Lev92,Lev95,Lev}, see also \cite{DL}) developed
	theory and proved universal upper bounds for $A_q(n,s)$. In this paper
	we describe refinements of the Levenshtein bound that can be applied
	for obtaining better bounds in the majority of the cases. Our
	refinements have two major advantages -- they are easy to derive
	and allow analytic investigation to certain extent.
	
	Improvements of the third Levenshtein bound in the binary case $q=2$
	were obtained by Tiet\"av\"ainen \cite{T1980} and Krasikov-Litsyn
	\cite{KL1997}, who developed bounds for $d=n-\frac{n+j}{2}$, where
	$0<j<\sqrt[3]{n}$.  Earlier, in 1973, linear programming bounds were
	obtained by McEliece (unpublished, see \cite[Chapter 17, Theorem
	38]{MWS}) who proved the asymptotic bound $A_2(n,s) \lesssim
	(n-j)(j+2)$, where $2d+j=n$ is as above, $s=1-\frac{2d}{n}$, and
	$j=o(\sqrt{n})$. The McEliece bound was improved in \cite{T1980} and
	\cite{KL1997} for $j=o(n^{1/3})$.
	
	On the other hand, binary codes of length $n=(2^{2m}+1)(2^m-1)$, size
	$M=2^{4m}$ and minimum distance $d=2^{2m-1}(2^m-1)$ (so $j=2^m-1$)
	were constructed by Sidel�nikov \cite{Sid}. This shows that the
	McEliece bound is of the correct order of magnitude. We also note that
	the maximum possible size of a code for $j=1$ and $n \equiv 1
	\pmod{4}$ is still unknown.
	
	Much less is known in the $q$-ary case, where analogs of the
	Tiet\"av\"ainen bound were obtained in \cite{Per}.  Our results give
	a generalization of the McEliece bound -- first as we
	prove it for every $n \geq q \geq 2$ and second, as we obtain its
	$q$-ary asymptotic analog.
	
	This paper is organized as follows. In Section 2 we explain the
	general linear programming bound, the Levenshtein bound and related
	parameters. Section 3 is devoted to general description of our
	refinements and discussion on its limits. We develop the first
	relevant case giving a rigorous proof for the refinement of the third
	Levenshtein bound in Section 4, where we also investigate the
	asymptotics of the new bounds. We also provide evidence that our
    bounds for large enough $\frac{d}{n}=\frac{1-s}{2}$ are as good as
    the complete linear programming despite being considerably simpler.
    Asymptotic bounds from the refinement
	of the fourth Levenshtein bound are presented in Section 5. We also
	compile a table of feasible parameters for good codes attaining
	our bounds.
	
	\section{Preliminaries}
	
	\subsection{Krawtchouk polynomials and the linear programming framework}

	For fixed $n$ and $q$, the (normalized) Krawtchouk polynomials are
	defined by
	\[
	Q_i^{(n,q)}(t) :=\frac{1}{r_i} K_i^{(n,q)}(d),
	\]
	where $d=\frac{n(1-t)}{2}$, $r_i=(q-1)^i {n \choose i}$, and
	$K_i^{(n,q)}(d)=\sum_{j=0}^{i}
	(-1)^j(q-1)^{i-j}\binom{d}{j}\binom{n-d}{i-j}$
	are the (usual) Krawtchouk polynomials that obey the three-term recurrent
	relation
	\[
	K_0^{(n,q)}(d)=1, \ \ K_1^{(n,q)}(d)=n(q-1)-qd,
	\]
	\[ K_{i+1}^{(n,q)}(d) =
	\frac{i+(q-1)(n-i)-qd}{i+1}K_i^{(n,q)}(d)
	-\frac{(q-1)(n-i+1)}{i+1}K_{i-1}^{(n,q)}(d),\mbox{ for } i\geq 1.
	\]
	We point out that even if the Krawtchouk polynomials  $K_{i}^{(n,q)}(d)$
    are defined for all non-negative integers $i$, the normalized polynomials
    $Q_i^{(n,q)}(t)$ are only defined for integers $i\in [0,n]$.
    If $f(t) \in \mathbb{R}[t]$ is a real polynomial of degree $m \geq 0$,
    then $f(t)$ can be uniquely expanded in terms of the Krawtchouk polynomials
    as $f(t) = \sum_{i=0}^m f_i Q_i^{(n,q)}(t)$.
	
	The next  (folklore) assertion is the main source of linear programming
	bounds (aka Delsarte bounds) for $A_q(n,s)$.
	
	\begin{theorem}
		\label{thm 1}
		Let $n \geq 2$ and $s \in [-1,1)$ be fixed and $f(t)$ be a real
		polynomial of degree $m$ such that:
		
		{\rm (A1)} $f(t) \leq 0$ for every $t \in T_n \cap [-1,s]$;
		
		{\rm (A2)} the coefficients in the Krawtchouk expansion $f(t) =
		\sum_{i=0}^{m} f_i Q_i^{(n,q)}(t)$ satisfy $f_i \geq 0$ for every $i$.
		
		Then  $A_q(n,s) \leq \frac{f(1)}{f_0}$.
	\end{theorem}
	
	\subsection{The Levenshtein bound}
	
	The so-called adjacent polynomials as introduced by Levenshtein
	(cf. \cite[Section 6.2]{Lev}, see also \cite{Lev92,Lev95}) are given
	by
	\begin{equation} \label{adjacent}
	Q_i^{(1,0,n,q)}(t) = \frac{K_i^{(n-1,q)}(d-1)}{\sum_{j=0}^i {n \choose
			j} (q-1)^j}, \ \ Q_i^{(1,1,n,q)}(t) =
	\frac{K_i^{(n-2,q)}(d-1)}{\sum_{j=0}^i {n-1 \choose j} (q-1)^j},
	\end{equation}
	\[
	Q_i^{(0,1,n,q)}(t) = \frac{K_i^{(n-1,q)}(d)}{{n-1 \choose i} (q-1)^i},
	\]
	where $d=n(1-t)/2$.
	
	For $a \in \{0,1\}$ and $i \in \{ 1,2,\ldots,n-1\}$, denote by
	$t_i^{1,a}$ the greatest zero of the adjacent polynomial
	$Q_i^{(1,a,n,q)}(t)$ (see \eqref{adjacent}) and also define
	$t_0^{1,1}=-1$ as well as $t_n^{1,0}=1$. We have the interlacing properties
    $t_{k-1}^{1,1}<t_k^{1,0}<t_k^{1,1}$, see \cite[Lemmas 5.29, 5.30]{Lev}.
    For a positive integer $m=2k-1+\varepsilon$, $\varepsilon \in \{0,1\}$, let
	\[
	\mathcal{I}_m :=
	\left[ t_{k-1+\varepsilon}^{1,1-\varepsilon},t_k^{1,\varepsilon} \right).
	\]
	Then the set of well defined intervals $\left\{\mathcal{I}_m\right\}_{m=1}^{2n-1}$
    forms a partition of the interval $[-1,1)$ into non-overlapping subintervals.
	For every $s \in \mathcal{I}_m$, Levenshtein used Theorem \ref{thm 1}
	with certain polynomials of degree $m$
	\begin{equation}
	\label{lev_poly}
	f_m^{(n,s,q)}(t)=(t-s)(t+1)^\varepsilon A^2(t)
	\end{equation}
	(see \cite[Equations (5.81) and (5.82)]{Lev}), where
	$\deg(f_m^{(n,s)})=m$, to obtain (see \cite[Equations (6.45) and
	(6.46)]{Lev})
	\begin{equation}
	\label{L_bnd}
	A_q(n,s) \leq L_{m}(n,s;q) = q^{1-\varepsilon}
	\left(
	1 - \frac{Q_{k-1}^{(1,1-\varepsilon,n,q)}(s)}{Q_k^{(0,\varepsilon,n,q)}(s)}
	\right)
	\sum\limits_{j=0}^{k-1+\varepsilon} {n \choose j}(q-1)^{j}
	\end{equation}
	for every $s\in \mathcal{I}_{m}$.  The bound \eqref{L_bnd} is attained
	by many codes with good combinatorial properties but is weak in many other
	particular cases. It is also worth to mention its good asymptotic
	behavior (see \cite{BN2006}, \cite[Section 6.2]{Lev}).
	
	In \cite{BD} two of the authors obtained (for any $q$) and investigated (for $q=2$)
	necessary and sufficient conditions for global optimality of the
	Levenshtein bounds (see also \cite[Theorem 5.47]{Lev}). Here we
	discuss another possibility of improving Levenshtein bounds by taking
	into account the discrete nature of the set of inner products.
	
	\section{Our refinement -- vanishing at inner products instead of zeros of
		the Levenshtein's polynomial}
	
	The roots $-1\leq \alpha_{0} < \alpha_{1} < \cdots <
	\alpha_{k-2+\varepsilon} < \alpha_{k-1+\varepsilon} = s$ of the Levenshtein
	polynomials $f_m^{(n,s,q)}(t)$ (we recall that $m=2k-1+\varepsilon$,
	$\varepsilon \in \{0,1\}$) are exactly the roots of the equation
	\begin{equation}
	\label{equ-roots}
	(t+1)^\varepsilon [P_k(t)P_{k-1}(s)-P_k(s)P_{k-1}(t)]=0,
	\end{equation}
	where $P_i(t)=Q_i^{(1,\varepsilon,n,q)}(t) \in \mathbb{Q}[t]$. Since
	\eqref{equ-roots} is equivalent to an equation with integer coefficients,
	the double zeros $\alpha_i$, $i=\varepsilon,\ldots,k-2+\varepsilon$, will
	rarely coincide exactly with inner products from the set $T_n$. Taking this
	into account we obtain the following refinement of the Levenshtein bound.
	
	We first locate the nodes $\alpha_i$, $i = \varepsilon, \ldots,
	k-2+\varepsilon$, with respect to the elements (the inner products) of
	$T_n$.  Then, if $\alpha_i \in (t_{j-1},t_{j})$ for some integer $j\in
	[1,n]$, we replace the double zero $\alpha_i$ by two simple
    zeros\footnote{In the special case $\alpha_i=t_{j}$ for some $i =
	\varepsilon, \ldots, k-2+\varepsilon$ and $j$ there are two possible
	replacements of $\alpha_i$ -- by $t_{j-1}$ and $t_{j}$ or by $t_{j}$ and
	$t_{j+1}$. Our choice is simple -- we check both and take the better one.}
	$\gamma_{2(i-\varepsilon) +1} = t_{j-1}$ and $\gamma_{2(i-\varepsilon +1)}
	=t_{j}$. After setting $\gamma_{2k-1} = s$, we define the polynomial
	\[
	f(t)=(t+1)^{\varepsilon}\prod_{i=1}^{2k-1}(t-\gamma_i)
	\]
	of degree $m=2k-1+\varepsilon$. We observe that the values of this
	polynomial in the interval $(t_{j-1},t_{j})$ are positive and, in
	particular, $f(\alpha_i) \geq 0$, with (very rare apart from
	$\alpha_{k-1+\varepsilon} = s \in T_n$) equality case if and only if
	$\alpha_i=t_{j}$. Finally, in the case when the degree $m$ exceeds the
	codelength $n$ we reduce the polynomial $f(t)$ to its remainder from its
    division by $\displaystyle g(t)=\prod_{i=0}^{n}(t-t_i)$.
    This operation is standard when the polynomial metric space (PMS) is finite.

	This construction clearly implies that the condition (A1) is
	satisfied. Moreover, using the quadrature formula
	\begin{equation}
	\label{f0-second}
	f_0 = \frac{1}{L_{m}(n,s;q)}+\sum_{i=0}^{k-1+\varepsilon} \rho_i f(\alpha_i)
	\end{equation}
	(see \cite[Theorem 5.39]{Lev}) and the inequalities $f(\alpha_i)
	\geq 0$ we conclude that $f_0>0$ always follows.
	
	The condition (A2) for $i \geq 1$ can be easily checked numerically in every
	particular case, and it is satisfied in the great majority of the cases we
	considered.	We give a rigorous proof for the case $m=3$ below.
	
	Summarizing, whenever we have (A2), Theorem \ref{thm 1} gives upper
	bound for the corresponding $A_q(n,s)$. Clearly, this is a strict
	improvement of the Levenshtein bound \eqref{L_bnd} if and only if
	$\alpha_i \not\in T_n$ for at least one $i$.  Note that $\alpha_i \in
	T_n$ occurs very rare -- this is connected to integral zeros of
	Krawtchouk polynomials (see \cite{KL1996,Stroeker}).
	
	Our numerical results cover wide range of values of $q$, $n$ and $s$,
	as we inspect all feasible $s$ for given $q$ and $n$.  Unfortunately,
	comparisons with well established sources such as \cite{Brouwer,AVZ,Bes} can
	be made in small range, namely, for alphabet size $q=2$, 3, 4, and 5, and
	for
	lengths $n \leq 28$, $n \leq 16$, $n \leq 12$, and $n \leq 11$,
	respectively.
	In these ranges we recover the following best known upper bounds
	\[
	A_3(14,-1/7) \leq 237, \ \ A_4(11,-3/11) \leq 320,
	\ \ A_5(11,-5/11) \leq 250
	\]
	(the Levenshtein bounds are 256.5, 364, and 265, respectively).
	
	It is clear that, in every particular case, the numerics from our
	refinements can not be better than the complete (integer) linear
	programming (see, for example \cite{sage}). However, we are going to
	show strong evidence that for every fixed $m$, our method gives the same
	results as the complete linear programming gives for large enough $n$
    despite being considerably simpler for computation.
	
    The much easier computation allows us to go for large lengths. Bounds for large lengths were
    numerically investigated (for binary
    codes only) by Barg-Jaffe \cite{BJ2001}. Our computational results agree well
    with their application of the simplex method for large $\frac{d}{n}=\frac{1-s}{2}$.
    We give a short table for comparison. The bounds are computed for $\frac{1}{n}\log_2 A_2(n,s)$.

\begin{center}
		{\bf Table 1.} Bounds for binary codes, $n=1000$, $\frac{d}{n}=\frac{1-s}{2} \in [0.25,0.45]$.

\smallskip

\begin{tabular}{|c|c|c|c|c|c|}
  \hline
  $d/n$ & 0.25 & 0.3 & 0.35 & 0.4 & 0.45 \\ \hline
  $L_{m}(n,s;2)$ & 0.387 & 0.283 & 0.191 & 0.115 & 0.505 \\
  our bound & 0.386 & 0.281 & 0.188 & 0.110 & 0.047\\
  simplex from \cite{BJ2001} & 0.380 & 0.280 & 0.188 & 0.109 & 0.047 \\
  \hline
\end{tabular}
\end{center}

\smallskip

	This comparison and our computational results for larger $q$ lead us to the conjecture that our
	method matches the best results possible by Theorem \ref{thm 1} for large enough ratio $d/n=(1-s)/2$.

\begin{conjecture}
		\label{conj-good-bound}
		For a fixed $q \geq 3$ there exist a constant $s_q$ such that whenever $s\in[-1,s_q)\cap T_n$ (that is large enough $d/n=(1-s)/2$) the
		above refinements are the best that can be obtained by Theorem \ref{thm 1}.
	\end{conjecture}

    In order to support this conjecture, in Figure \ref{fig1} we present graph for the function $s_q(n)$ defined as the range $[-1,s_q(n))\cap T_n$ where our improvement is optimal in the sence of Theorem \ref{thm 1}.

	\begin{figure}[h!]
		\centering
		\psfrag{s}     []  []  [0.6]{$s_{q}(n)$}
		\psfrag{n}     []  []  [0.6]{$n$}
		\psfrag{MyTitle}   []  []  [0.4]{}
		\psfrag{q3}    []  [r]  [0.6]{$\ \ \ \ \ q=3$}
		\psfrag{q5}    []  [r]  [0.6]{$\ \ \ \ \ q=5$}
		\psfrag{q10}   []  [r]  [0.6]{$\ \ \ \ \ q=10$}
		
		\includegraphics[width=1\textwidth]{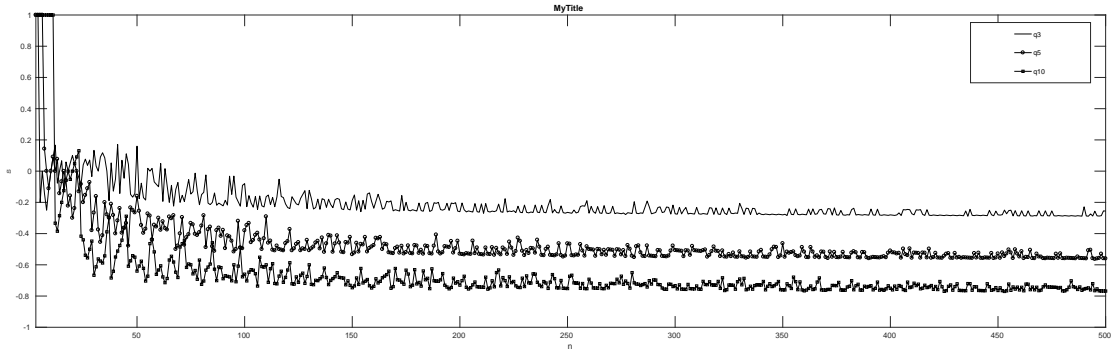}

		\caption{The function $s_{q}(n)$ for some $q$'s and codelengths $n\in
			[3,500]$.}
		\label{fig1}
		
	\end{figure}

	We note also that, similarly to \cite{BJ2001}, our computations do not support
	the conjecture of Samorodnitsky \cite{Sam} on the exponential strength of the linear programming method.

	Semidefinite programming was shown to be better than the linear
	programming in most particular cases (see, for example, \cite{GST,Sch,LPS})
	but it hardly gives bounds in analytic forms. So we see another advantage in
	giving analytic form of our bounds in the first cases (see the next
	sections).
	
	\section{The refinement of $L_3(n,s;q)$}
	
	In this section we apply our refinement in the case of the third Levenshtein
	bound. We provide proof for the feasibility of the chosen polynomial as
	well as some numerical results for the global optimality of this polynomial.
	
	\subsection{Proof of the feasibility of suggested polynomial}
	
	Let us set\footnote{This convention is natural extension of McEliece's $d=(n-j)/2$
    used for $q=2$.}
	\begin{equation}
	\label{d-asymp}
	d = n-1 - \frac{n-2+j}{q},
	\end{equation}
	where the parameter $j$ will be explained below. We proceed with the
	general case of upper-bounding $A_q(n,s)$ in the range
	\begin{equation*}
	s \in \mathcal{I}_3 = \left[ t_1^{1,1},t_2^{1,0} \right) =
	\left[
	-\frac{(q-2)(n-2)}{nq}, -\frac{(q-2)(n-2)}{nq}+\frac{S_1-q}{nq}
	\right),
	\end{equation*}
	where $S_1=\sqrt{q^2+4(q-1)(n-2)}$.  Since $d=n(1-s)/2$, we obtain the
	ranges for $d$ and $j$ to be
	\begin{equation}
	\label{D3}
	d \in
	\left(
	(n-1) - \frac{n-2}{q} - \frac{S_1-q}{2q}, (n-1) - \frac{n-2}{q} \,
	\right]
	\iff
	j \in \mathcal{J}_3 = \left[0,\frac{S_1-q}{2}\right).
	\end{equation}
	The simple form of $\mathcal{J}_3$ justifies the change of the variable.
	
	In the particular case of $m=3$ the polynomial defined in \eqref{lev_poly}
	becomes
	\begin{equation*}
	f_3^{(n,s,q)}(t)=(t-\alpha_0)^2(t-s)=(t-\alpha_0)^2(t-\alpha_1)\, ,
	\end{equation*}
	where $s = \alpha_1 = 1 - \frac{2d}{n} = -1 + \frac{2(n-2+j+q)}{nq}$ and
	$\alpha_0 = -1 + \frac{2j(n-1)}{nq(j+q-1)}$, since $\alpha_0$ and $s$
	are the roots of the equation \eqref{equ-roots} for $k=2$ and
	$\varepsilon=0$.  Let us set $d_0 = \frac{n(1-\alpha_0)}{2} =
	n-\frac{j(n-1)}{q(j+q-1)}$ and define $e$ to be the unique rational number
	in the interval $(0,1]$ such that $d_0+e$ is integer. We point out  that
	\begin{equation*}
	\alpha_0 \in
	\left(
	1 - \frac{2(d_0+e)}{n}, 1-\frac{2(d_0+e-1)}{n}\,
	\right]
	\end{equation*}
	and that $e=\frac{r}{q(j+q-1)}$, where $r$ is the positive remainder from
	the division of $j(n-1)$ by $q(j+q-1)$, i.e. $r\equiv  j(n-1)
	\pmod{q(j+q-1)}$ and $r\in \left( 0,q(j+q-1) \right]\cap \mathbb{Z}$.
	
	Now we are in a position to define our improving polynomial as
	\begin{equation}
    \label{newpoly3}
	\begin{array}{rcl}
	f(t) & = & \left(t-1+\frac{2(d_0+e)}{n}\right)
	           \left(t-1+\frac{2(d_0+e-1)}{n}\right)(t-s) \\
	     & = & \left(t+1+\frac{2e}{n}-\frac{2j(n-1)}{nq(j+q-1)}\right)
	           \left(t+1+\frac{2(e-1)}{n}-\frac{2j(n-1)}{nq(j+q-1)}\right)
	           \times \\ &   &
               \left(t+1-\frac{2(n-2+j+q)}{nq}\right) \\
	     & = & f_0+f_1Q_1^{(n,q)}(t)+f_2Q_2^{(n,q)}(t)+f_3Q_3^{(n,q)}(t),
	\end{array}	\end{equation}
	where the coefficients $f_i$, $i=3,2,1,0$, are given by
	\begin{eqnarray}
	f_3 &=& \frac{8(q-1)^3(n-2)(n-1)}{q^3n^2}>0, \nonumber \\
	f_2 &=& \frac{8(q-1)^2(n-1)A}{q^3n^2(q+j-1)}, \label{f2} \\
	f_1 &=& \frac{8(q-1)\left( (eq-B)^2+C \right)}{q^3n^2}, \label{f1} \\
	f_0 &=& \frac{8\left(a^2(2-q-j)+Da+E\right)}{q^3n^3}>0, \nonumber
	\end{eqnarray}
	for
	\begin{eqnarray*}
		A &=& -j^2+(2eq-1)j+(q-1)(2n+2eq+q-4),\\
		B &=&  \frac{1}{j+q-1}\left[j(j-2)-n(q-1)+\frac{q}{2}(3j+q-1)\right],\\
		C &=& -j^2+(-q+2)j+(3n-2)(q-1)-\frac{q^2}{4},\\
		D &=& (j+q-1)[2n(q-1)-q]+q,\\
		E &=& -n(n-1)(q-1)^2(j+q),\\
		a &=& \frac{(n-1)(q-1)(q+j)}{q+j-1}+eq.
	\end{eqnarray*}
	
	We proceed with the proof of the positivity of the coefficients $f_1$
	and $f_2$.
	
	\begin{lemma}
		\label{f2positive}
		We have $f_2>0$ for every $n \geq 2$ and $q \geq 2$.
	\end{lemma}
	\begin{proof}
		It follows from \eqref{f2} that it is enough to prove that $A>0$.
		Observing that $A$ is a concave quadratic function in $j$, according
		to (\ref{D3}) we have to check the positivity of $A$ for $j=0$ and
		$j=\frac{S_1-q}{2}$, respectively. For the former we have
		\[
		A=(q-1)(2n+2eq+q-4)\geq (q-1)(2n+q-4)>0
		\]
		and for the latter
		\[
		A=\left( eq+\frac{q-1}{2}\right) S_1+eq(q-2)+\frac{(2n+q-4)(q-1)}{2}>0,
		\]
		whenever $n\geq 2$ and $q\geq 2$.
	\end{proof}
	
	\begin{lemma}
		\label{f1positive}
		We have $f_1>0$ for every $n\geq q\geq 2$.
	\end{lemma}
	
	\begin{proof}
		According to \eqref{f1} it is sufficient to show that $C > 0$. As in
		the proof
		of Lemma \ref{f2positive}, we observe that $C$ is a concave quadratic
		function
		of $j$ so we check its values at the limits of the interval
		$\mathcal{J}_3$.
		We
		obtain those to be
		\[
		(3n-2)(q-1)-\frac{q^2}{4} \geq (3q-2)(q-1)-\frac{q^2}{4}
		= 2 + \frac{q}{4}(11q-20) > 0
		\]
		and
		\[
		S_1+2n(q-1)-q-q^2/4 \geq S_1+\frac{q}{4}(7q-12) > 0,
		\]
		whenever $n\geq q\geq 2$.
	\end{proof}
	
	\begin{remark}
		Our numerics suggest that we might always have $f_2>f_1>f_0$.  Since
		$f_0>0$ follows by the formula \eqref{f0-second}, another proof of the
		positivity of $f_1$ and $f_2$ could probably be done along these
		lines.
	\end{remark}
	
	\begin{theorem}
		\label{refin3}
		We have
		\begin{equation}
		\label{new-bound-3}
		A_q(n,s) \leq \frac{a(a+q)dq}{a^2(2-q-j)+Da+E},
		\end{equation}
		where the parameters are determined as above.
	\end{theorem}
	
	\begin{proof}
		The condition (A1) is obviously satisfied by our improving
		polynomial. The condition (A2) is satisfied as well -- Lemmas
		\ref{f2positive} and \ref{f1positive} give $f_2 \geq 0$ and $f_1 \geq
		0$, respectively, $f_3>0$ is obvious, and $f_0>0$ follows, as
		mentioned above, from \eqref{f0-second}. Therefore $A_q(n,s) \leq
		\frac{f(1)}{f_0}$ and simplifications give the desired bound.
	\end{proof}

	\begin{example}
		\label{e1-320}
		For $q=4$, $n=11$ and $s=-3/11$ (this $s$ corresponds to minimum
		distance $d=n(1-s)/2=7$) we are in the range of the third Levenshtein
		bound $L_3(11,s;4)$. Since $\alpha_0=-\frac{17}{22} \in
		\left(-\frac{9}{11},-\frac{7}{11}\right)$ and
		$\alpha_1=s=-\frac{3}{11}$, we have our improving polynomial as
		follows:
		\begin{eqnarray*}
			f(t) &=&			
			\left(t+\frac{9}{11}\right)\left(t+\frac{7}{11}\right)
			\left(t+\frac{3}{11}\right)
			\\
			&=& \frac{63}{5324} Q_0^{(11,4)}(t)+\frac{117}{484}
			Q_1^{(11,4)}(t)+\frac{45}{44} Q_2^{(11,4)}(t)+\frac{1215}{484}
			Q_3^{(11,4)}(t)
		\end{eqnarray*}
		and $A_4(11,-3/11) \leq \frac{f(1)}{f_0}=320$ (here
		$L_3(11,-3/11;4)=364$).  The best known lower bound in this case is $A_4
		(11,-3/11) \geq 128$ (see \cite{Brouwer}).  Further analysis via the
		distance distributions of a putative quaternary
		$(11,320,-3/11)$ code $C$ does not give a contradiction. Indeed, all
		possible inner products of $C$ are $-\frac{3}{11}$, $-\frac{7}{11}$
		and $-\frac{9}{11}$ and such a code must be distance regular,
		i.e. every point of $C$ has the same distance distribution, which
		turns out to be integral.
	\end{example}
	
	We now calculate the asymptotic form of the bound from Theorem
	\ref{refin3}.
	
	\begin{corollary}
		\label{McE-q}
		Let $j=cn^\alpha\in \mathcal{J}_3$ for some positive constant $c$ and
		some $\alpha\in [0,1/2]$.  The behavior of the upper bound
		given by (\ref{new-bound-3}) as $n \to \infty$ is as follows
		\begin{equation}
		\label{asy-3-1}
		A_q(n,s) \leq [(q-1)n-(j+q-2)](j+q)+j(j+q-1)^2 + o(1),\  \alpha \in
		[0,1/5),
		\end{equation}
		\begin{equation}
		\label{asy-3-2}
		A_q(n,s) \leq (q-1)(q+j)n+j^3+\frac{c^5n^{5\alpha-1}}{q-1}+o(n),\
		\alpha \in
		[1/5,1/2),
		\end{equation}
		and
		\begin{equation}
		\label{asy-3-3}
		A_q(n,s) \leq \frac{c(q-1)^2}{q-1-c^2}n^{3/2}+
		\frac{(q-1)(c^4-(q-1)(3c^2-q^2+q))}{(q-1-c^2)^2}n +o(n),\  \alpha= 1/2.
		\end{equation}
	\end{corollary}
	
	\begin{proof}
		The upper bound in \eqref{new-bound-3} can be re-written as
		\begin{equation}
		\begin{array}{lcl}
		A_q(n,s) & \leq & [(q-1)n-(j+q-2)](j+q)+j(j+q-1)^2 \\
		&      &   +
		\displaystyle \frac{[(q-1)n+(eq+1)(eq+1-q)]j^6 +P_{q,e,n}(j)}
		{(q-1)^2(j+q)n^2+Q_{q,e,n}(j)},
		\end{array}
		\label{new-bound-3j}
		\end{equation}
		where $P_{q,e,n}(j)$ and $Q_{q,e,n}(j)$ are polynomials in $j$
		of degrees $5$ and $3$, respectively, and with coefficients
		that are linear in $n$. We notice that for the fraction in
		\eqref{new-bound-3j}, the nominator is of order
		$\Theta(n^{1+6\alpha})$ and that the denominator has the order
		$\Theta(n^{2+\alpha})$ for all $\alpha\in[0,1/2]$. This gives the
		result in \eqref{asy-3-1} since $1+6\alpha<2+\alpha$ when $\alpha \in
		[0,1/5)$.

        To obtain \eqref{asy-3-2}, we observe that for large $n$ the
		nominator behaves like $(q-1)nj^6 + \Theta(n^{1+5\alpha})$. For all
		$\alpha\in[1/5,1/2)$ we have $1+5\alpha-(2+\alpha)<1$ and also
		$i\alpha < 1$, for $i=0,1,2$. Now \eqref{asy-3-2} follows by
		ignoring the terms of orders in $n$ that are less than
		$1$.

        Finally, by substituting $j=c\sqrt{n}$ in \eqref{new-bound-3j}
		and performing polynomial division for the polynomials in the
		variable $x=\sqrt{n}$ we arrive at \eqref{asy-3-3}.
	\end{proof}
	

	\subsection{On Conjecture \ref{conj-good-bound} for $m=3$}
	
	We reformulate the linear programming bound back to its
	classical form (see \cite[Sections 3.2 and 3.3]{Del}, \cite[Section
	3B]{DL}, \cite[Corollary 2.7]{Lev95}).
	
	\begin{theorem}
		\label{Th_LPbound_orig}
		Let the polynomial $g(z)=\sum_{i=0}^{n} g_i K_i^{(n,q)}(z)$ satisfies
		the conditions
		\[ g_0>0, \ g_i \geq 0, \ i = 1,2,\ldots,n; \]
		\[ g(0)>0, \ g(i) \leq 0, \ i = d,d+1,\ldots,n. \]
		Then $A_q(n,s) \leq g(0)/g_0$, where $s=1-2d/n$.
	\end{theorem}
	
	In the light on Theorem \ref{Th_LPbound_orig} the best upper bound on
	the quantity $A_q(n,s)$ is obtained by the polynomial
	$g^*(z)=1+\sum_{i=1}^{n} \frac{x_i^*}{r_i} K_i^{(n,q)}(z) = 1 +
	\sum_{i=1}^{n} \frac{K_i^{(n,q)}(z)}{(q-1)^i{n \choose i}}x_i^*$ for
	which the coefficients $\overline{x}^*=(x_1^*,x_2^*,\cdots,x_n^*)$
	constitute a solution to the linear optimization problem
	
	\begin{equation}
	\label{LP_problem}
	\begin{array}{lrll}
	\text{minimize}
	\displaystyle & x_1+x_2+ \cdots + x_n & & \\
	\text{subject to} &
	\displaystyle
	\sum_{l=1}^{n} \frac{K_l^{(n,q)}(i)}{r_l} \!\!\! & x_{l} \leq -1,
	& i = d, (d+1),\dots , n\\
	& & x_{l} \geq 0,  & l=1,2,\dots ,n.
	\end{array}
	\end{equation}
	
	Applying the KKT optimality conditions (see for example \cite[Section
	5.5]{BV04}) we can conclude that necessary and sufficient condition
	for $\overline{x}^*$ to be optimal is the existence of numbers
	$\lambda_l,\, l=1,2,\dots,n$, and $\mu_i,\, i=d,d+1,\dots,n$, such
	that
	\begin{equation}
	\label{KKTcond}
	\begin{array}{rcll}
	\lambda_l & = &
	\displaystyle
	1 + \sum_{i=d}^{n} \mu_i\frac{K_l^{(n,q)}(i)}{r_l} \geq 0,
	&  l=1,2,\dots,n,\\
	\lambda_lx_l^* & = & 0, & l=1,2,\dots,n, \\
	\mu_ig^*(i)      & = & 0, & i=d,(d+1),\dots,n, \\
	\mu_i          & \geq & 0, & i=d,(d+1),\dots,n.
	\end{array}
	\end{equation}
	
	Equation \eqref{KKTcond} turns out to be a very powerful tool for
	checking the global optimality of a given polynomial. In particular,
	if we have a polynomial $f(t)$ of degree $m$ that satisfies
	conditions (A1) and (A2) of Theorem \ref{thm 1} and has $m$
	different roots in the interval $[-1,s]$, then we can exactly
	determine the numbers $\mu_i,\,i=d,d+1,\dots,n$, if the Krawtchouk
	expansion of $f(t)$ in (A2) has strictly positive coefficients. The
	polynomial $f(t)$ would then be globally optimal if and only if all the
	lambdas, $\lambda_l,\, l=1,2,\dots,n$, calculated as in \eqref{KKTcond} are
	non-negative. Our approach for improving the Levenshtein bound very often
	results in such polynomials.
	
	Let us now consider the polynomial $f(t)$ as defined by \eqref{newpoly3}
	and let us set $g(z)=\frac{1}{f_0}f\left(1-\frac{2z}{n}\right)$. We can
	easily verify that
	\begin{equation}
	\label{g3_poly}	
	g(z) = 1 + \sum_{i=1}^{3} \frac{f_i}{f_0} \frac{K_i^{(n,q)}(z)}{r_i}.
	\end{equation}
	We now determine the Lagrange multipliers $\lambda_i^*$ and $\mu_i^*$  for
	the
	polynomial defined in \eqref{g3_poly}. It has already been shown that
	$f_i>0,\,
	i=0,1,2,3$, which according to \eqref{KKTcond} means that $\lambda_i^*=0,\,
	i=1,2,3$. Since $g(i)=0$
	only for $i\in \{d,d_0+e-1,d_0+e\}$ we have $\mu_i^*=0$ for all $i\in
	\{d+1,d+2,\dots,n\}\setminus\{d_0+e-1,d_0+e\}$. The remaining
	three $\mu_i^*$'s can be obtained from the system of linear equations
	\begin{equation}
	\label{mui_system}
	\mu_d^*\frac{K_l^{(n,q)}(d)}{r_l} +
	\mu_{d_0+e-1}^*\frac{K_l^{(n,q)}(d_0+e-1)}{r_l} +
	\mu_{d_0+e}^*\frac{K_l^{(n,q)}(d_0+e)}{r_l}  = -1,\,  l=1,2,3.
	\end{equation}
	The system \eqref{mui_system} has an unique solution $(\mu_d^*,
	\mu_{d_0+e-1}^*,\mu_{d_0+e}^*)$ with help of which we can calculate the
	remaining $\lambda_i^*$, for $i=4,5,\dots,n$, according to
	\begin{equation}
	\label{lambda_calc}
	\lambda_i^*= 1+\mu_d^*\frac{K_i^{(n,q)}(d)}{r_i} +
	\mu_{d_0+e-1}^*\frac{K_i^{(n,q)}(d_0+e-1)}{r_i} +
	\mu_{d_0+e}^*\frac{K_i^{(n,q)}(d_0+e)}{r_i}.
	\end{equation}
	The first step towards the calculation of the lambdas is the following
	statement.
	\begin{lemma}
		\label{weights3}
		The  weights  $(\mu_d^*, \mu_{d_0+e-1}^*,\mu_{d_0+e}^*)$ that solve
		system
		\eqref{mui_system} can be calculated as
		\begin{eqnarray*}
			\mu_d^* &=& \frac{n(q-1)CD[C+q(j+q-1)]}{AB[B+q(j+q-1)]},  \\
			\mu_{d_0+e-1}^* & = &\frac{en(q-1)E[C+q(j+q-1)]}{AB},\\
			\mu_{d_0+e}^* & = & \frac{(1-e)n(q-1)CE}{A[B+q(j+q-1)]},
		\end{eqnarray*}
		where
		\begin{eqnarray*}
			A &=& -(q+j-2)[eq(j+q-1)]^2 + (n-1)(q-1)(j+q)[n(q-1)-j(q+j-2)]\\
			& & + eq(j+q-1)[(q^2+jq-q-2j)(q+j-2)+2n(q-1)],\\
			B &=& (n-2+j)(q-1)+j(j-1)+eq(j+q-1),\\
			C &=& (n-1)(q-1)(j+q)+eq(j+q-1), \\
			D &=&\left			
			[(n-1)(q-1)+(2e-1)(j+q-1)\frac{q}{2}\right]^2 +
			(j+q-1)^2\left[(n-1)(q-1)-\frac{q^2}{4}\right],\\
			E &= & (j+q-1)^3[(n-1)(q-2)+n-j].\\
		\end{eqnarray*}
	\end{lemma}
	\begin{proof}
		Direct check shows that  the above defined $\mu_d^*, \mu_{d_0+e-1}^*$
		and $\mu_{d_0+e}^*$ satisfy \eqref{mui_system} for any $n$, $q$,
		$e\in(0,1]$, $j\in\mathcal{J}_3$ and $l=1,2,3$.
	\end{proof}
	
	The non-negativity of $\mu_d^*, \mu_{d_0+e-1}^*$ and $\mu_{d_0+e}^*$ for
	$n\geq q\geq 2$, $j\in\mathcal{J}_3$ and any $e\in[0,1)$ can be derived
	from Lemma \ref{weights3} by showing the positivity of the  parameters
	$A,B,C,D$ and $E$. Obviously $B>0$ and $C>0$ with the only exception of the
	trivial case $n=q=2$, $e=j=0$ for which $B=0$. The parameter $D$ is
	positive since $(n-1)(q-1)\geq (q-1)^2 \geq (q/2)^2$ whenever $n\geq q \geq
	2$ with equality only for $n=q=2$. As $A$ is a quadratic function in $e$
	with negative leading coefficient, its positivity for  $e\in[0,1)$ can be
	checked by investigating the values for $e=0$ and $e=1$. For these values
	we have  $(n-1)(q-1)(j+q)[n(q-1)-j(q+j-2)]$ and
	$[(n(q-1)+q+1)j+(n+1)q(q-1)][n(q-1)-j(q+j-2)]$, respectively. For any
	$j\in\mathcal{J}_3$ we have
	\[
	n(q-1)-j(q+j-2)\geq n(q-1)-\frac{S_1-q}{2}(q+\frac{S_1-q}{2}-2)
	= (q-2)+S_1 > 0,
	\]
	which shows the positivity of $A$. Finally, the positivity of $E$ follows
	from the fact that  $(n-1)(q-2)+n>(S_1-q)/2\geq j$.

	We summarize the above observations into the following result.

	\begin{theorem}
	\label{test3-theorem}
	Let $f(t)$ be the third degree polynomial given in
	\eqref{newpoly3} and let $\lambda_i^*$, for $i=4,5,\dots,n$, be
	given by \eqref{lambda_calc}, where the triple
	$(\mu_d^*,\mu_{d_0+e-1}^*,\mu_{d_0+e}^*)$ is defined as the unique solution
	to the linear equation system \eqref{mui_system}. Then if
	$\lambda_i^*\geq 0$ for every integer $i\in [4,n]$, the bound
	\eqref{new-bound-3} on $A_q(n,s)$ is the best one that can be obtained by
	the linear programming method described in Theorem \ref{thm 1}.
	\end{theorem}
	
	The above statement is a powerful tool for checking the global optimality
	of the suggested polynomial in the case of the third Levenshtein bound. A
	similar result can be obtained for the cases when the bound of higher order
	is valid. However, in those cases the non-negativity of the $\mu_j^*$'s is
	not always true and thus has to be added to the non-negativity
	condition on the $\lambda_i^*$'s. Some observations in this directions are
	provided in the next section.
	
	Our numerical results suggest that Theorem \ref{test3-theorem} is
	applicable in all cases with very few exceptions. We have been able to
	verify that for codelengths $n$ up to $1000$ and alphabet sizes in the
	range $3\leq q \leq 10$, the only cases when the suggested polynomial does
	not	provide the optimal linear programming bound are for $q=3$ and
	$n\in\{5,7,8,9\}$.
	
	\section{Refinements of $L_4(n,s;q)$ and $L_5(n,s;q)$}
	
	The Levenshtein bound $A_q(n,s) \leq L_4(n,s;q)$ is valid in the range
	\begin{equation*}
        s \in \mathcal{I}_4 = \left[ t_2^{1,0},t_2^{1,1} \right) =
        \left[\frac{S_1-(q-2)(n-2)-q}{nq},
        \frac{S_2-(q-2)(n-3)}{nq} \right),
    \end{equation*}
	where $S_1$ is as above and $S_2=\sqrt{q^2 + 4(q-1)(n-3)}$, $n \geq 3$.
	Then
	\begin{eqnarray*}
		&& d \in \mathcal{D}_4 =
		\left( n-1-\frac{n-3}{q}-\frac{q+S_2}{2q}, n-1-\frac{n-2}{q}-
		\frac{S_1-q}{2q}
		\right] \\
		&\iff& j \in
		\mathcal{J}_4=\left[ \frac{S_1-q}{2},\frac{S_2+q}{2}-1\right).
	\end{eqnarray*}
	
	The polynomial from (\ref{lev_poly}) is
	\[	
	f_4^{(n,s,q)}(t)=(t+1)(t-\alpha_1)^2(t-s)=(t-\alpha_0)(t-\alpha_1)^2(t-\alpha_2),
	\]
	where $s = \alpha_2 = 1-\frac{2d}{n} = -1 + \frac{2(n-2+j+q)}{nq}$ again,
	$\alpha_1 = -\frac{(n-2)(j(q-2)+2(q-1))}{nqj}$, and $\alpha_0=-1$
	($\alpha_1$ and $s$ are the roots of the equation \eqref{equ-roots} for
	$k=2$
	and $\varepsilon=1$).
	
	We set $d_0 = \frac{n(1-\alpha_1)}{2} =
	n-1-\frac{(j-q+1)(n-2)}{qj}$ and define $e$ to be the unique rational number
	in the interval $(0,1]$ such that $b:=d_0+e$ is integer.
	Then our improving polynomial is
	\begin{equation}
	\label{newpoly4}
	\begin{array}{rcl}
	f(t) &=& (t+1)\left(t-1+\frac{2b}{n}\right)
	\left(t-1+\frac{2(b-1)}{n}\right)(t-s) \\
	&=& f_0+f_1Q_1^{(n,q)}(t)+
	f_2Q_2^{(n,q)}(t)+f_3Q_3^{(n,q)}(t)+f_4Q_4^{(n,q)}(t),
	\end{array}
	\end{equation}
	The positivity of the coefficients $f_1$, $f_2$ and $f_3$ can be approached
	like in the previous section but we prefer to omit the cumbersome
	calculations and to go directly to an asymptotic.
	
	\begin{theorem}
		\label{refin4}
		Provided $f_i \geq 0$ for $i=1,2,3,4$, we have
		\[
		A_q(n,s) \leq \frac{q^3b(b-1)(n(q-1)-j-q+2)}
		                   {(1-j)q^2b^2 + C_1qb - C_2},
		\]
		where $b$ and $j$ are determined as above, $C_1 = j(q-1)(2n-1)+j-q$,
		and $C_2 = (q-1)(n-1)[(q-1)(j+1)n+2(j-q+1)]$.
	\end{theorem}
	
	\begin{proof}
		Under the assumptions, the polynomial $f(t)$ satisfies the conditions of
		Theorem
		\ref{thm 1}.
		Thus it is enough to compute $f(1)$ and $f_0$ and
		to plug in $f(1)/f_0$.
	\end{proof}
	
	We are not aware of improvements of the fourth Levenshtein bound in the spirit of the discussion from the previous section.
    We proceed with an analog of the McEliece bound. The interval $\mathcal{J}_4$ is short and we can express $j$ as
	$j=\frac{S_1-q}{2}+c$, where $c\in\left[0,(q-1)\left(1-\frac{2}{S_1+S_2}\right)\right)$ is some constant. Note that $c \in [0,q-1)$.
	
	\begin{theorem}
		\label{McE-q1}
		For any $s=\frac{S_1-n(q-2)+2c+q-4}{nq}\in T_n$ and $c\in\left[0,(q-1)\left(1-\frac{2}{S_1+S_2}\right)\right)$ we have
		\begin{equation}
		\label{asy-4}
		A_q(n,s) \lesssim\frac{q(q-1)^2n^2}{2(q-c)}
		\ as \
		n
		\to \infty .
		\end{equation}
	\end{theorem}
	
	\begin{proof}
		For large $n$ and $c\in\left[0,(q-1)\left(1-\frac{2}{S_1+S_2}\right)\right)$, we have
		\[
		f_4>0, \
		f_3 \sim \frac{16(q-1)^3}{q^4}>0, \
		f_2 \sim \frac{16(q-1)^{2.5}}{q^4n^{0.5}}>0, \
		f_1 \sim \frac{32(q-1)^2}{q^4n}>0.
		\]
		Therefore (A1) and (A2) are satisfied and $A_q(n,s) \lesssim f(1)/f_0$.
		The calculation of the asymptotic of $f(1)/f_0$ now gives \eqref{asy-4}.
	\end{proof}
	
	The analytical investigation of the refinement of the fifth Levenshtein bound
	$L_5(n,s;q)$ seems technically quite difficult. It is convenient, however, to illustrate the computational strength of our method --
    we are able to reach lengths  $10000$ (for $q=3$) in about $30$ minutes of computations on an Intel Core2 Duo P9300 @ 2.26GHz processor. For any fixed $n$ we compute all bounds in the range of
	$L_5(n,s;3)$, which amounts to $225654$ cases in the codelength range $6\leq n\leq 10000$. The computations include verification of the fact that
    $f_i\geq 0$ for $i=1,2,3,4,5$.  With no exception, the requirement (A2) in Theorem \ref{thm 1} has been satisfied.

	Finally, we note that the refinement of $L_5(n,s;q)$ is attained
	asymptotically (since the Levenshtein bound $L_5(n,s;q)$ itself is attained)
	by the Kerdock codes \cite{Ker72} of length $n=2^{2\ell}$,
	cardinality $M=n^2=2^{4\ell}$
	and minimum distance $d=(n-\sqrt{n})/2=2^{2\ell-1}-2^{\ell-1}$.

	
	\section{Parameters of putative codes attaining our bounds}
	
	In the table below, we list all codes which would attain, if exist, our
	refinement of the third Levenshtein bound
	$L_3(n,s)$, in the range $n \leq 100$ for the lengths and $2 \leq q \leq 5$
	for the alphabet size. The bound $L_3(n,s)$ is shown in the fourth column.
	The sixth column contains the roots of our polynomials, i.e. the only three
	possible inner products of attaining codes, and the last column gives the
	distance distribution of such codes (ordered accordingly to the inner
	products). The cases where the best known upper bound from \cite{Brouwer} is
	repeated are marked with asterisk.

	The putative optimal codes must be 3-designs and this allows one to compute their
	distance distribution. Of course, if the distance distributions is not integral,
	such code does not exist. For lengths $n \leq 300$, there are 7 out of 38 (for $q=2$),
	14 out of 54 (for $q=3$), 20 out of 47 (for $q=4$), and 18 out of 39 (for $q=5$)
	cases which pass the integrality test. Extended version of the table will be uploaded on the Internet.
	
	\begin{center}
		{\bf Table 2.} Parameters for attaining the refinement of
		$L_3(n,s)$,
		$n \leq 100$, $2 \leq q \leq 5$

		{\small
			\label{tab:1}
			\smallskip
			
			\begin{tabular}{|c|c|c|c|c|c|c|}
				\hline
				$q$ & $n$ & $d$ & $L_3(n,s)$ & Refinement & Inner products           & Distance distribution  \\ \hline
                 2  &  12 &  5  &  62.50     &        60  & $-1/2, -1/3, -1/6$       &    $5, 15, 39$         \\ \hline
				 2  &  56 & 25  &  1135      &      1100  & $-5/28, -1/7, 3/28$      &    $175, 275, 649$     \\ \hline
				 2  &  90 & 41  &  2863.69   &      2788  & $-2/15, -1/9, 4/45$      &    $492, 697, 1598$    \\ \hline
				 2  &  96 & 45  &  1161      &      1155  & $-1/6, -7/48, 1/16$      &    $90, 252, 812$      \\ \hline
				*3  &   4 &  2  &    33      &        27  & $-1, -1/2, 0$            &    $6, 8, 12$          \\ \hline
                 3  &   7 &  4  &    57      &        54  & $-1, -5/7, -1/7$         &    $4, 14, 35$         \\ \hline
				 3  &  20 & 12  &   312.429  &       306  & $-7/10, -3/5, -1/5$      &    $16, 85, 204$       \\ \hline
				 3  &  25 & 15  &   531      &       513  & $-3/5, -13/25, -1/5$     &    $114, 75, 323$      \\ \hline
				 3  &  27 & 16  &   874      &       840  & $-5/9, -13/27, -5/27$    &    $272, 84, 483$      \\ \hline
				 3  &  40 & 24  &  2421      &      2349  & $-1/2, -9/20, -1/5$      &    $928, 144, 1276$    \\ \hline
				 3  &  52 & 32  &  2094      &      2052  & $-1/2, -6/13, -3/13$     &    $608, 208, 1235$    \\ \hline
				 3  &  88 & 55  &  5745      &      5670  & $-5/11, -19/44, -1/4$    &    $1925, 440, 3304$   \\ \hline
                 4  &   4 &  2  &    83.20   &        64  & $-1, -1/2, 0$            &    $21, 24, 18$        \\ \hline
				*4  &   5 &  3  &    76      &        64  & $-1, -3/5, -1/5$         &    $18, 15, 30$        \\ \hline
                 4  &   8 &  5  &   182.50   &       160  & $-1, -3/4, -1/4$         &    $15, 60, 84$        \\ \hline
                 4  &   9 &  6  &   136      &       128  & $-1, -7/9, -1/3$         &    $16, 27, 84$        \\ \hline
				*4  &  11 &  7  &   364      &       320  & $-9/11, -7/11, -3/11$    &    $99, 55, 165$       \\ \hline
				 4  &  13 &  9  &   196      &       192  & $-1, -11/13, -5/13$      &    $9, 39, 143$        \\ \hline
				 4  &  18 & 12  &   697.6    &       640  & $-7/9, -2/3, -1/3$       &    $135, 144, 360$     \\ \hline
				 4  &  42 & 30  &  1190.59   &      1184  & $-16/21, -5/7, -1/7$     &    $36, 259, 888$      \\ \hline
				 4  &  49 & 35  &  1660      &      1640  & $-5/7, -33/49, -1/7$     &    $205, 245, 1189$    \\ \hline
				 4  &  56 & 39  &  7676.5    &      7176  & $-9/14, -17/28,	-11/28$  &    $1287, 2093, 3795$  \\ \hline
                *5  &   4 &  2  &   167.86   &       125  & $-1, -1/2, 0$            &    $52, 48, 24$        \\ \hline
                *5  &   5 &  3  &   191.67   &       125  & $-1, -3/5, -1/5$         &    $44, 40, 40$        \\ \hline
				*5  &   6 &  4  &   145      &       125  & $-1, -2/3, -1/3$		 &    $44, 24, 60$        \\ \hline
                 5  &   9 &  6  &   485      &       375  & $-1, -7/9, -1/3$         &    $44, 162, 168$      \\ \hline
				*5  &  11 &  8  &   265      &       250  & $-1, -9/11, -5/11$		 &    $40, 44, 165$       \\ \hline
				 5  &  16 & 12  &   385      &       375  & $-1/ -7/8, -1/2$		 &    $30, 64, 280$       \\ \hline
				 5  &  21 & 16  &   505      &       500  & $-1, -19/21, -11/21$     &    $16, 84, 399$       \\ \hline
				 5  &  25 & 18  &  3621      &      3645  & $-19/25, -17/25, -11/25$ &    $1638, 132, 1694$   \\ \hline
				 5  &  45 & 34  &  3649      &      3250  & $-7/9, -11/15, -23/45$   &    $429, 792, 2028$    \\ \hline
				 5  &  55 & 42  &  3705.8    &      3675  & $-43/55, -41/55, -29/55$ &    $132, 1078, 2464$   \\ \hline
				 5  &  72 & 56  &  3257.26   &      3250  & $-29/36, -7/9, -5/9$     &    $64, 585, 2600$     \\ \hline
				 5  &  75 & 57  & 12141      &     11970  & $-53/75, -17/25, -39/75$ &    $4617, 608, 6744$   \\ \hline
				 5  &  91 & 70  &  9725      &      9625  & $-5/7, -9/13, -49/91$    &    $2695, 780, 6149$   \\ \hline
				 5  &  92 & 70  & 26339.3    &     25025  & $-16/23, -31/46, -12/23$ &    $7084, 4784, 13156$ \\ \hline
				 5  & 100 & 76  & 55841      &     55195  & $-17/25, -33/50, -13/25$ &    $26809, 912, 27473$ \\ \hline
			\end{tabular}
		}
	\end{center}

\end{document}